\documentclass[10pt,twocolumn,twoside]{IEEEtran}

\newcommand{\E}{{\rm I\!E}}

\usepackage{amsmath}
\usepackage{amssymb}
\usepackage{amsfonts}

\usepackage{amsthm}
\usepackage{overpic}
\usepackage{lipsum}
  {
      \theoremstyle{plain}
      \newtheorem{assumption}{Assumption}
  }

\usepackage[utf8]{inputenc}
\usepackage[T1]{fontenc}
\usepackage{url}
\usepackage{ifthen}
\usepackage{cite}
\usepackage{graphicx,epstopdf,amssymb,amsthm}
\usepackage{ulem,color}
\usepackage{epstopdf}
\usepackage[utf8]{inputenc}
\usepackage[english]{babel}
\usepackage{caption}
\let\labelindent\relax
\usepackage{enumitem}

\DeclareCaptionLabelFormat{lc}{\MakeLowercase{#1}~#2}
\newtheorem{theorem}{Theorem}
\newtheorem{corollary}{Corollary}
\newtheorem{remark}{Remark}
\newtheorem{lemma}{Lemma}

\newtheorem{definition}{Definition}
\newtheorem{proposition}{Proposition}
\newtheorem{example}{Example}[]

\newcommand{\phil}[1]{\textcolor{blue}{#1}}

\begin{document}
\graphicspath{{figs/}}

%
\title{Estimation and Distributed Eradication of SIR Epidemics on Networks}
%
%

%

\author{Ciyuan Zhang, Humphrey Leung, Brooks Butler, 
        and~Philip. E. Par\'e*
\thanks{*Ciyuan Zhang, Humphrey Leung, Brooks Butler,
 and~Philip. E. Par\'e are with the School of Electrical and Computer Engineering at Purdue University. Emails: \{zhan3375, leung61, brooksbutler, philpare\}@purdue.edu.}\thanks{*This work was funded in part by the C3.ai Digital Transformation Institute sponsored by C3.ai Inc. and the Microsoft Corporation and in part 
   by the National Science Foundation, grants NSF-CNS \#2028738 and NSF-ECCS \#2032258.}}

\maketitle

\begin{abstract}
This work examines the discrete-time networked SIR (susceptible-infected-recovered) epidemic model, where the infection and recovery parameters may be time-varying. We provide a sufficient condition for the SIR model to converge to the set of healthy states exponentially. We propose a stochastic framework to estimate the system states from observed testing data and provide an analytic expression for the error of the estimation algorithm. Employing the estimated and the true system states, we provide two novel eradication strategies that guarantee at least exponential convergence to the set of healthy states. We illustrate the results via simulations over northern Indiana, USA.
\end{abstract}

\IEEEpeerreviewmaketitle

\section{Introduction}

As of February 2021, the COVID-19 virus has claimed 2.4 million lives and infected 110 million individuals worldwide~\cite{whoCoronavirus}. Lack of effective treatments, high contagion rates~\cite{mohapatra2020recent}, long incubation periods~\cite{backer2020incubation, guan2020clinical, li2020early, lauer2020incubation}, and asymptomatic cases~\cite{byambasuren2020estimating, chang2020time, mizumoto2020estimating,ing2020covid} pose significant challenges in containing and eradicating pandemics. Recent pandemics, including gonorrhea~\cite{lajmanovich1976deterministic}, Ebola~\cite{dike2017susceptible}, and COVID-19~\cite{calafiore2020modified}, have accelerated the development of infection models. The main goal of epidemic model development is to identify conditions to eradicate the pathogen, and leverage the knowledge of these conditions to design mitigation strategies. Various infection models have been proposed, based on characteristics of individual pathogens, and studied in the literature, including susceptible-infected-susceptible (SIS), susceptible-infected-removed (SIR), and susceptible-infected-removed-susceptible (SIRS)~\cite{rock2014dynamics, mei2017dynamics}. 
In this paper, we focus on the SIR epidemic model. We aim to expand on the SIR model, by exploring mutating viruses over networks, estimation of the underlying states, and distributed eradication strategies.

The patchwork response to COVID-19~\cite{haffajee2020thinking} gives rise to susceptible community subpopulations, with heterogeneous time-varying factors not previously explored by the SIR model. Extensions on the SIS model, studied in~\cite{van2009virus, ahn2013global, liu2019networked}, augment the compartmental epidemic models originated in~\cite{bailey1975mathematical} to include interactions between subpopulations of susceptible communities.  Additionally, various advanced epidemic models consider 
time-varying factors~\cite{
pascual2005seasonal, gracy2019mutating, pare2015stability, pare2018epidemic, bokharaie2010spread, gracy2020asymptotic, prakash2010virus, liu2016threshold}. In this paper, we establish sufficient conditions for the set of healthy states of a networked time-varying SIR model to be globally exponentially stable (GES). These equilibrium states are not unique, as the final susceptible and removed states are dependent on the initial conditions and the time-varying infectious and healing parameters.

The delay in onset of COVID-19 symptoms~\cite{backer2020incubation, guan2020clinical, li2020early, lauer2020incubation}, large asymptomatic populations estimated between $17-81\%$~\cite{byambasuren2020estimating, chang2020time, mizumoto2020estimating, ing2020covid}, and delay in test results~\cite{bergquist2020covid} compromise the ability for accurate estimation of current infection states.
An estimation algorithm that incorporated a constant delay between the change in infection proportion and testing data was introduced in \cite{hota2020closed}. They studied the inference problem by using a Bayesian approach. Inspired by the delay characterization suggested in \cite{hota2020closed}, we propose a stochastic delay to model the unpredictability of the COVID-19 virus and testing strategies. We have developed methods for estimating the underlying epidemic states from testing data with a delay sampled from a geometric distribution, which cannot be completely filtered by the method suggested in \cite{hota2020closed}. The geometric delay model accounts for the stochastic effect of individuals failing to get tested immediately after exposure. We study the aggregated effect of each individual delay on the trajectory of confirmed cases and devise a method for estimating the underlying epidemic states of an SIR model from these delayed measurements. We also investigate the proposed method's estimation error, which provides insights for achieving an accurate estimation of the system states. We then employ this more realistic estimation to strategically eradicate a disease.

As proven in~\cite{hota2020closed}, the SIR epidemic model converges to a healthy state, however, an exponential convergence is not shown. Combining the modeling and inference approach allows us to develop two distributed control strategies capable of eradicating epidemic spread exponentially, at an equilibrium with a higher proportion of susceptible population. Decreasing the removed (recovered) and increasing the susceptibility proportion is of exceptional importance for the COVID-19 pandemic, as long term and severe health complications have been documented in the recovered populations, including impaired cognition~\cite{liotta2020frequent} and damage to cardiac tissue~\cite{mitrani2020covid}. 
Our main result shows that by applying the estimated and the true susceptible states of each node, our proposed eradication strategies will guarantee global exponential stability of a healthy state of the overall network. 

\subsection{Paper Contributions}
We summarize the main contributions of this paper as follows:
\begin{itemize}
    \item We establish sufficient conditions for global exponential stability of the set of healthy states; see Theorem~\ref{thm:GES}. 
    \item We propose a stochastic framework which estimates the trajectories of the system states of the networked SIR model from testing data. Furthermore, we provide analytical expressions for the error of the estimation algorithm we propose; see Prop.~\ref{prop:estimation_error}.
    \item We propose two distributed eradication strategies for adjusting healing rates, one that is based on the true system states and the other one is based on the inferred system states. Both methods guarantee that the virus is eradicated within exponential time; see Theorem~\ref{thm:control_datadriven} and Corollary~\ref{coro:control_inferred}. 
\end{itemize}

\subsection{Paper Outline}
We organize this paper as follows: 
Section \ref{sec:model} lays down some basic assumptions and restates the well-known SIR model in the networked fashion, 
and it presents the main problems studied in this paper. 
Section \ref{sec:stability} first recalls preliminary results that are essential for stability analysis and then discusses the sufficient conditions for global exponential stability of a healthy state of the networked time-varying SIR models.
Section~\ref{sec:inference} covers the proposed techniques of estimating hidden epidemics states with the stochastic delay of tested individuals and testing data.
Section \ref{sec:eradication_strategies} covers the two distributed control strategies which ensure that the system converges to a healthy state in at least exponential time.
Section~\ref{sec:simulation} illustrates the results from Section \ref{sec:inference} and \ref{sec:eradication_strategies} with numerical simulations. Finally, in Section~\ref{sec:conclusion}, we summarize the main conclusions of this paper and discuss potential future directions. 


\subsection{Notation}

We denote the set of real numbers, the non-negative integers, and the positive integers as $\mathbb{R}$, $\mathbb{Z}_{\geq 0}$, and $\mathbb{Z}_{\geq 1}$, respectively. For any positive integer $n$, we have $[n] =\{1,2,...,n \}$. The spectral radius of a matrix $A \in \mathbb{R}^{n\times n}$ is $ \rho(A)$. A diagonal matrix is denoted as diag$(\cdot)$. The transpose of a vector $x\in \mathbb{R}^n$ is $x^\top$. The Euclidean norm is denoted by $\lVert \cdot \rVert$. We use $I$ to denote identity matrix. We use $\mathbf{0}$ and $\mathbf{1}$ to
denote the vectors whose entries all equal 0 and 1, respectively. The dimensions of the vectors are determined by context. Given a matrix $A$, $A \succ 0$ (resp. $A \succeq 0$) indicates that $A$ is positive definite (resp. positive semidefinite), whereas $A \prec 0$ (resp. $A \preceq 0$) indicates that A is negative definite (resp. negative semidefinite). Let $G =(\mathbb{V},\mathbb{E})$ denote a graph or network where $\mathbb{V} = \{ v_1, v_2,..., v_n\}$ is the set of subpopulations, and $\mathbb{E} \subseteq \mathbb{V}\times \mathbb{V} $ is the set of edges. We denote the expectation of a random variable as $ \E[\cdot]$.


\section{Model and Problem Formulation}\label{sec:model}
Consider a time-varying epidemic network of $n$ subpopulations, where the size of subpopulation $v_i$ is $N_i \in \mathbb{Z}_{>0}$, and 
the infection rates and healing rates could be time-varying. We denote $\beta_{ij}(t)$ as the infection rate from node $v_j $ to node~$v_i$ at time $t$, we denote $\gamma_i(t)$ as the healing rate of node~$v_i$ at time $t$. The proportions of the subpopulation at node~$v_i$ 
which are susceptible, infected, and recovered at time $t$ are denoted by $s_i(t),x_i(t)$ and $r_i(t)$, respectively. The deterministic continuous-time evolution of the SIR epidemic is given by
\begin{subequations}
\IEEEyesnumber\label{eq:ct_SIR} 
\begin{align}
  \dot{s}_i(t)&= -s_i(t) [\sum_{j=1}^n \beta_{ij}(t)x_j(t)], \label{eq:ct_SIRsub1}\\
 \dot{x}_i(t) &= s_i(t)[\sum_{j=1}^n \beta_{ij}(t)x_j(t)] -\gamma_i(t) x_i(t),  \label{eq:ct_SIRsub2}
\\
\dot{r}_i(t) &= \gamma_i(t) x_i(t), \;\ \forall i \in [n].
\label{eq:ct_SIRsub3}
\end{align}
\end{subequations}
We now state the discrete-time SIR epidemic dynamics obtained through Euler discretization of~\eqref{eq:ct_SIR}. For a small sampling time $h>0$, the discrete-time evolution of the SIR epidemic is given by
\begin{subequations}
\IEEEyesnumber\label{eq:dt_SIR} 
\begin{align}
  s_i(k+1) &= s_i(k) +h[-s_i(k)\sum_{j=1}^n \beta_{ij}(k) x_j(k)], \label{eq:dt_SIRsub1}\\
 x_i(k+1) &= x_i(k) +  h[s_i(k)\sum_{j=1}^n \beta_{ij}(k) x_j(k)-\gamma_i(k)x_i(k)],  \label{eq:dt_SIRsub2}
\\
 r_i(k+1) &= r_i(k) +h\gamma_i(k) x_i(k). \label{eq:dt_SIRsub3}
\end{align}
\end{subequations}


Equation \eqref{eq:dt_SIRsub2} can be rewritten as 
\begin{equation}\label{eq:SIR_dynamic}
    x(k+1) =x(k) +h [S(k)B(k) - \Gamma(k)]x(k),
\end{equation}
where $S(k) = \text{diag}(s(k))$, $B(k)$ is the matrix with $(i,j)$th entry $\beta_{ij}(k)$, and $\Gamma(k) = \text{diag}(\gamma_i(k))$. The spread of a virus over a network can be captured using a graph $G= (\mathbb{V},\mathbb{E})$, 
where $\mathbb{E} = \{ (v_i, v_j) | \beta_{ij}(k) \neq 0 \}$ is the set of directed edges. 

We make the following assumptions 
in order for
the system in~\eqref{eq:dt_SIR} 
to be
well defined.
\begin{assumption}\label{assume:one}
    For every $i\in [n]$, $h \gamma_i(k) > 0 $ and $\forall j\in [n], \beta_{ij}(k) \geq 0$, for every $k\in \mathbb{Z}_{\geq 0}$.
\end{assumption}
\begin{assumption}\label{assume:two}
    For every $i\in [n]$, $h \gamma_i(k) \leq 1 $ and $h \sum_j \beta_{ij}(k) \leq 1$, for every $k\in \mathbb{Z}_{\geq 0}$.
\end{assumption}
We have the following result which shares the same idea as the time-invariant model, proven in~\cite{hota2020closed}.
\begin{lemma}\label{lemma:one}
Suppose $s_i(0), x_i(0), r_i(0) \in [0,1]$, $s_i(0)+x_i(0)+r_i(0) =1$, and Assumptions~\ref{assume:one} and~\ref{assume:two} hold. Then, for all $k \in \mathbb{Z}_{\geq0}$,
\begin{enumerate}
    \item $s_i(k), x_i(k), r_i(k) \in [0,1]$, 
    \item $s_i(k)+x_i(k)+r_i(k) =1$, and
    \item $s_i(k+1) \leq s_i(k)$.
\end{enumerate}
\end{lemma}

\begin{definition}
We define the set of healthy states of~\eqref{eq:dt_SIR} as $\{s_i^*(k), x_i^*(k), r_i^*(k): i\in [n], k \in \mathbb{Z}_{\geq 0}\}$, where $x_i^*(k)=0$, $s_i^*(k)\in [0,1],$ and $  r_i^*(k)\in [0,1]$ for all $i \in [n]$.
\end{definition}

Given a network that is infected by a virus, our goal is to guarantee that each subpopulation $v_i$ converges to the set of healthy states in exponential time regardless of the initial conditions of the each subpopulations. 
We now officially state the questions being studied in this paper:
\begin{enumerate}[label=(\roman*)]
    \item \label{itm:first} For the system with dynamics given in~\eqref{eq:SIR_dynamic}, under what condition is the set of healthy states, i.e., 
    $x(k)=\mathbf{0}$,
    global exponentially stable (GES)?
    \item \label{itm:second_a} Given the testing data, how can the stochastic framework be constructed to estimate
    the susceptible,
    infected,  
    and recovered proportions, 
    denoted by
    $\widehat{s_i}(k)$, $\widehat{x_i}(k)$, and $\widehat{r_i}(k)$, respectively, 
    for each subpopulation $v_i$ in the network?
    \item \label{itm:second_b} What is the estimation error of the stochastic framework that we proposed?
    \item \label{itm:third_a} Given the knowledge of the conditions that ensure GES of a healthy state, i.e., 
    $x(k)=\mathbf{0}$ and 
    $\widehat{s_i}(k)$ inferred from testing data, 
    how can we devise dynamic control algorithms which apply new healing rates $\widehat{\gamma_i}(k)$ to each agent in~\eqref{eq:SIR_dynamic} so that the epidemic is eradicated with a faster rate of convergence than the rate of exponential?
\end{enumerate}

\section{Stability Analysis} \label{sec:stability}

This section 
presents
conditions that ensure
global exponential stability of the set of healthy states. 
First, we introduce some preliminaries and then we present our main analysis results.

\subsection{Preliminaries}
In this subsection, we 
recall 
results that are crucial for understanding the rest of the paper.

\begin{lemma}\label{lemma:diagonal}
\cite{rantzer2011distributed}
Suppose that $M$ is a nonnegative matrix which satisfies $\rho(M)<1$. Then there exists a diagonal matrix $P \succ 0$ such that $M^\top P M -P \prec 0$. 
\end{lemma}
Consider a system described as follows:
\begin{equation}\label{eq:lip}
    x(k+1) = f(k,x(k))\phil{.}
\end{equation}
\begin{definition}\label{def:GES}
An equilibrium point of \eqref{eq:lip} is GES is there exist positive constants $\alpha$ and $\omega$, with $0\leq \omega <1$, such that
\begin{equation}
    \lVert x(k) \rVert \leq \alpha \lVert x(k_0) \rVert \omega^{(k-k_0)}, \forall k, k_0 \geq 0, \forall x(k_0) \in \mathbb{R}^n.
\end{equation}
\end{definition}
We recall a sufficient condition for GES of an equilibrium of \eqref{eq:lip} 
from~\cite{vidyasagar2002nonlinear}.
\begin{lemma}\label{lemma:GES}
\cite[Theorem 28]{vidyasagar2002nonlinear}
Suppose there exists a function $V: \mathbb{Z}_+ \times \mathbb{R}^n \rightarrow \mathbb{R}$, and constants $a,b,c>0$ and $p>1$ such that $a\lVert x \rVert^p \leq V(k,x) \leq b \lVert x \rVert^p$, $\Delta V(k,x):= V(x(k+1))-V(x(k)) \leq -c\lVert x \rVert^p, \forall k \in \mathbb{Z}_{\geq0}$, and $\forall x (k_0)\in \mathbb{R}^n$, then $x(k)=\mathbf{0}$ is a globally exponential stable equilibrium of \eqref{eq:lip}. 
\end{lemma}
\begin{lemma}\label{lemma:rate_GES}
\cite[Theorem 23.3]{rugh1996linear} Under the assumption of Lemma~\ref{lemma:GES}, the rate of convergence to 
the origin 
is upper bounded by an exponential rate of $\sqrt{1-(c/b)} \in [0,1)$, where $b$ and $c$ are defined in Lemma~\ref{lemma:GES}.
\end{lemma}

Note that a healthy state of the system in \eqref{eq:dt_SIR} is GES if Assumptions \ref{assume:one} and \ref{assume:two} hold and the condition in Definition~\ref{def:GES} and Lemma~\ref{lemma:GES} are satisfied for all $x(k_0) \in [0,1]^n$, since this is the domain where the model is well defined.


\subsection{Global Exponential Stability of the Healthy States}
In this subsection, 
we present sufficient conditions for the global exponential stability of the set of healthy states of the system. We find the conditions by analyzing the spectral radius of the state transition matrix of~\eqref{eq:dt_SIRsub2}. 
We define
\begin{align}
     M(k) &= I-h\Gamma(k) +hB(k), \label{eq:M} \\
     \hat{M}(k)  &= I+ h[S(k)B(k) - \Gamma(k)].
    \label{eq:Mhat}
\end{align}
Notice that $\hat{M}(k)$ is the state transition matrix of~\eqref{eq:dt_SIRsub2} and it can be written that
\begin{equation}\label{eq:M_M_hat}
    \hat{M}(k) = M(k)-h(I-S(k))B(k).
\end{equation}
We use $M(k)$ and $\eqref{eq:M_M_hat}$ to illustrate the sufficient conditions for the GES of the set of healthy states in the subsequent theorem.
\begin{theorem}\label{thm:GES}
Given Assumptions \ref{assume:one} and \ref{assume:two}, suppose for all $k \in \mathbb{Z}_{\geq0}$, $B(k)$ is symmetric. 
If $\text{sup}_{k\in \mathbb{Z}_{\geq 0}} \rho(M(k))< 1$, 
then the set of healthy states of \eqref{eq:dt_SIR} is GES.
\end{theorem}
\begin{proof}
See Appendix. 
\end{proof}

Recall from the previous result that $M(k)$ is a nonnegative matrix which satisfies $\text{sup}_{k\in \mathbb{Z}_{\geq 0}} \rho(M(k))< 1$, such that $M^\top(k) Q(k+1) M(k) - Q(k)\prec 0$, where $Q(k)$ is a diagonal matrix defined in the Lyapunov function:
\begin{equation}
    V(k,x) = x^\top Q(k)x.
\end{equation}
\begin{corollary}\label{coro:rate_exp}
Under the assumptions of Theorem~\ref{thm:GES}, the rate of convergence to a healthy state is upper bounded by an exponential rate of
$\sqrt{1-\frac{\sigma_3}{\sigma_2}}$, where $\sigma_2=\max_{k\in \mathbb{Z}_{\geq 0}} \lambda_{\text{max}}(Q(k))$, $\sigma_3 =\max_{k\in \mathbb{Z}_{\geq 0}} \lambda_{\text{min}}[Q(k)-M(k)^\top  Q(k+1)M(k)]$.
\end{corollary}
\begin{proof}
See Appendix. 
\end{proof}
\begin{remark}
Notice that in Theorem~\ref{thm:GES}, $\text{sup}_{k\in \mathbb{Z}_{\geq 0}} \rho(M(k))< 1$ is the key condition which ensures the set of healthy states of~\eqref{eq:dt_SIR} is GES. We can interprete $\rho(M)$ in the context of epidemiology as the basic reproduction number of the virus over the network. Particularly, Theorem~\ref{thm:GES} affirms that given the time-varying parameters of the network satisfy the condition provided, 
the mutating virus will exponentially converge to the set of healthy states. 
\end{remark}
In this section, we found the conditions that ensure exponential convergence to the set of healthy states of~\eqref{eq:dt_SIRsub2} in Theorem~\ref{thm:GES}. This answers question~\ref{itm:first} in Section~\ref{sec:model}. 
The stability condition can help the policymakers reallocate the medical resources, staff etc. which leads to modifying the parameters in~\eqref{eq:dt_SIR} so that the spreading of the virus stops completely. 
One of the other factors that will assist in decision making is the COVID-19 observed testing data.

\section{State Estimation from Testing Data} \label{sec:inference}

In this section, we study how to estimate the epidemic states $(s(k), x(k), r(k))$ from testing data in order to design a feedback controller in the following section. One of the challenges of estimating the underlying system states is that the testing data on a given day does not capture the new infections on the same day. Instead, the testing data is a delayed representation of the change in the system. Characterizing the delay of each individual is difficult, because the delay is determined by numerous factors such as the incubation period of COVID-19, the duration of obtaining test results, the willingness of each individual to get tested, etc. Therefore, we propose a stochastic framework in this section to capture the factors which cause the testing delay.

\begin{definition}
The testing delay $\tau_i$ is the length of time between when an individual from subpopulation $v_i$ is infected and when their positive test result is reported.
\end{definition}
In our discrete-time model, we assume that $\tau_i \in \mathbb{Z}_{\geq 0}$.
We model the testing delay of each infected individual $\tau_i$ as two aggregate components to represent the uncertainty in the testing process:
\begin{equation}\label{eq:testing_delay}
    \tau_i =\eta_i+\mathcal{Y}_i, 
\end{equation}
where $\eta_i\in \mathbb{Z}_{\geq 0}$ is a constant and $\mathcal{Y}_i$ is sampled from a discrete time random variable 
whose measurable space is~$\mathbb{Z}_{\geq 0}$.
\begin{remark}
In~\eqref{eq:testing_delay}, the constant component $\eta_i$ can be interpreted as the length of time needed to acquire testing results. 
The random variable can be interpreted as the incubation period and/or the amount of time that it takes an individual to get tested after becoming infected. 
\end{remark}


\begin{figure}
\centering
\begin{overpic}[width = \columnwidth]{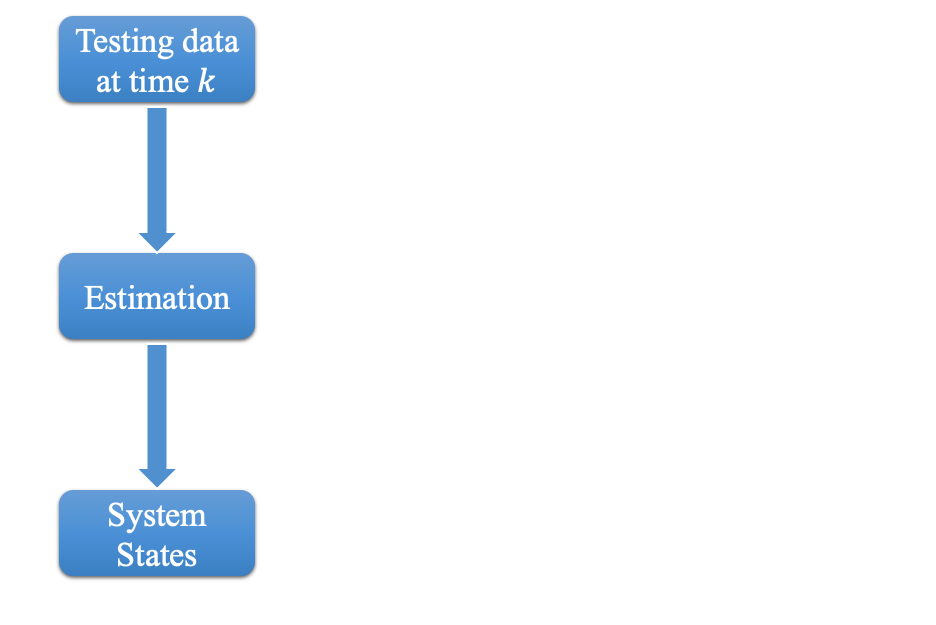}
        \put(21,47){{\parbox{0.75\linewidth}{\small
     $\Omega_i(k) = \big(C_i(k), D_i(k)\big)$ $\rightarrow \big( c_i(k), d_i(k)\big)$
     }}}\normalsize
     \put(21,22){{\parbox{0.75\linewidth}{\small
     $\widehat{\Theta}_i(k-\tau_i)=\big(\widehat{s_i}(k-\tau_i),\widehat{x_i}(k-\tau_i),\widehat{r_i}(k-\tau_i)\big)$}}}\normalsize
     \put(21,1){{\parbox{0.75\linewidth}{\small
     $\widehat{\Theta}_i(k)=\big(\widehat{s_i}(k),\widehat{x_i}(k),\widehat{r_i}(k)\big)$}}}\normalsize
      \end{overpic}
\caption{Estimation of System States from Testing Data 
}
\label{fig:Diagram_Inference}
\end{figure}

First, we denote the set of estimated system states for subpopulation $v_i$ at time $k$ as $\widehat{\Theta}_i(k)=\big(\widehat{s_i}(k),\widehat{x_i}(k),\widehat{r_i}(k)\big)$, we denote the set of testing data recorded at time $k$ to be $\Omega_i(k)= \big(C_i(k), D_i(k)\big)$, where $C_i(k)$ is the number of confirmed cases at time $k$, and $D_i(k)$ represents the number of removed (recovered) cases at time $k$. In addition, the cumulative number of confirmed and removed cases at node $v_i$ are written as $\mathbb{C}_i(k) = \sum_{j=0}^k C_i(j)$ and $\mathbb{D}_i(k) = \sum_{j=0}^k D_i(j)$, respectively. Therefore, the number of active cases is calculated by $\mathbb{A}_i(k) =\mathbb{C}_i(k) - \mathbb{D}_i(k)
$.
Recall that the size of each subpopulation is $N_i$; we define $c_i(k)= \frac{C_i(k)}{N_i}$ and $d_i(k)= \frac{D_i(k)}{N_i}$ as the proportion of daily confirmed cases and removal, respectively. Note that $c_i(k), d_i(k) \in [0,1]$, for all $i\in[n], k\in \mathbb{Z}_{\geq 0}$.
The estimation procedure is illustrated in Fig.~\ref{fig:Diagram_Inference}.

We then study how to relate $c_i(k)$ to the underlying states. We define a vector space $\Pi_{T_1}$ 
as the space of all the proportions of daily number of confirmed cases from time step $k=T_1$ to time step $k=T_2+1$. 
We define $\Xi_{T_1}$ as the vector of all the decreases in the proportion of susceptible individuals, $-\Delta s_i(k)$, from time step $k=T_1$ to time step $k=T_2+1$.
We denote $\Phi(T_1, T_2)$ as the transfer matrix which results in
\begin{equation}\label{eq:transfer_matrix}
    \Pi_{T_1} = \Phi(T_1, T_2) \Xi_{T_1},
\end{equation}
where $\Phi(T_1, T_2)$ is a $(T_2-T_1+2) \times (T_2-T_1+2)$ matrix, which depends on the SIR dynamics, the testing strategies, and the delay.


When the testing delay is a constant, i.e., $\tau_i =\eta_i > 0$, the only non-zero entries in $\Phi(T_1, T_2)$ are: $ \Phi_{l+\eta_i, l} =1, l\in [T_2-T_1+2-\eta_i]$.
Since for all $k \in [T_1+ \eta_i,T_2+1 ]$, we can write $c_i(k)$ as
   \begin{equation}
       c_i(k) = -\Delta s_i(k-\eta_i).
   \end{equation}
When the delay $\eta_i =0$, the transfer matrix $\Phi(T_1, T_2)=I$. 
Because for all $k \in [T_1,T_2+1 ]$, we can write that
      \begin{equation}
       c_i(k) = -\Delta s_i(k).
   \end{equation}

We now propose a stochastic testing framework to capture the delay between when an individual is infected and when they receive a positive test result. We first let $\eta_i =0$ in~\eqref{eq:testing_delay}, without the loss of generality. Furthermore, we assume that each infected individual at node $v_i$ has an equal probability $p_i^x \in (0,1]$ 
of receiving a diagnostic test each day starting from the day after they are infected. 
Therefore, we model $\mathcal{Y}_i$ in~\eqref{eq:testing_delay} as a random variable following the geometric distribution, 
with the probability of an infected individual acquiring a positive test $\delta$ days after infection being:
\begin{equation}\label{eq:geo_dist_single}
    P(\mathcal{Y}_i =\delta) =p_i^x(1-p_i^x)^{\delta -1}
\end{equation}
for $\delta \in \mathbb{Z}_{\geq 1}$.
The geometric distribution of the testing delay models the number of days before an infected individual obtains a diagnostic test which represents the incubation period of COVID-19 and/or the unwillingness of each individual getting a test. We assume that the delay of each infected individual's distribution is i.i.d. (independent and identically distributed) from others. Furthermore, we assume that an infected individual can be tested only once. Based on~\cite{an2020clinical} and~\cite{abbasi2020promise}, we assume that even if an individual recovers from COVID-19, their antibody tests will still give positive results. We also assume that all the tests generate accurate results. 

We now relate the proportion of confirmed cases $c_i(k)$ with the underlying states of the system. We define a binary random variable $\mathcal{X}_i(\nu)$ with $\mathcal{X}_i(\nu )=1$ (resp. $\mathcal{X}_i(\nu) =0$) if a randomly chosen individual from subpopulation $v_i$ became infected at time $\nu$. It can be written that 
\begin{equation} \label{eq:infection_nu}
\mathcal{X}_i (\nu) = \begin{cases} 
1 \ \  w.p. \ \ -\Delta s_i(\nu) \\
0 \ \  w.p. \ \ 1+\Delta s_i(\nu),
\end{cases}
\end{equation}where, from~\eqref{eq:dt_SIR}, $-\Delta s_i(\nu) =s_i(\nu-1)-s_i(\nu)=hs_i(\nu-1)\sum_{j}\beta_{ij}x_j(\nu-1)\geq 0$ for all $\nu \geq 0$.

We define the binary random variable $\mathcal{T}_i(\mu,\delta)$, with $\mathcal{T}_i(\mu,\delta)=1$ 
if a randomly chosen individual acquired a positive test at time $\mu$ and was infected  $\delta$ days before $\mu$.
Now we rewrite $\nu $, in \eqref{eq:infection_nu}, as $\mu - \delta$. 
From~\eqref{eq:geo_dist_single}, the conditional probability $P(\mathcal{T}_i(\mu,\delta)=1 |\mathcal{X}_i(\mu - \delta)=1)$
is given by the geometric probability mass function (pmf) $p_i^x(1-p_i^x)^{\delta-1}$ and represents the probability of an infected individual acquiring a positive test specifically $\delta$ days after being infected. 
Hence, the joint pmf of the two random variables $\mathcal{X}_i(\mu -\delta) , \mathcal{T}_i(\mu,\delta) $ is written as:
\begin{equation}
    P_{\mathcal{X}_i,\mathcal{T}_i}(\mu -\delta,\mu)  = P(\mathcal{X}_i(\mu -\delta) 
    \cap 
    \mathcal{T}_i(\mu,\delta)),
\end{equation}
which is interpreted as the probability that
a randomly chosen individual 
became 
infected 
at time $\mu -\delta$ and 
acquired 
a positive test at time $\mu$, where $\mu -\delta , \mu \in [T_1, T_2]$.


Therefore, the joint pmf $P_{\mathcal{X}_i,\mathcal{T}_i}(\mu-\delta,\mu)$ is calculated as:
\begin{align}
    & P_{\mathcal{X}_i, \mathcal{T}_i}  (\mathcal{X}_i(\mu-\delta)=1 \cap \mathcal{T}_i(\mu,\delta)=1) \nonumber  \\ &= P(\mathcal{T}_i(\mu,\delta)=1 |\mathcal{X}_i(\mu - \delta)=1)P(\mathcal{X}_i(\mu -\delta)=1) \nonumber  \\ &= p_i^x(1-p_i^x)^{\delta-1}[-\Delta s_i(\mu- \delta)], \label{eq:joint_pmf}
\end{align}
\begin{align}
    & P_{\mathcal{X}_i, \mathcal{T}_i}  (\mathcal{X}_i(\mu-\delta)=0 \cap \mathcal{T}_i(\mu,\delta)=1) \nonumber  \\ &= P(\mathcal{T}_i(\mu,\delta)=1 |\mathcal{X}_i(\mu - \delta)=0)P(\mathcal{X}_i(\mu -\delta)=0) \nonumber  \\ &=0[1+\Delta s_i(\mu-\delta)]= 0, \label{eq:zero}
\end{align}
since we assume that the test results are accurate. Similarly,
\begin{align}
    & P_{\mathcal{X}_i, \mathcal{T}_i}  (\mathcal{X}_i(\mu-\delta)=1 \cap \mathcal{T}_i(\mu,\delta)=0) \nonumber  \\ &= P(\mathcal{T}_i(\mu,\delta)=0 |\mathcal{X}_i(\mu - \delta)=1)P(\mathcal{X}_i(\mu -\delta)=1) \nonumber  \\ &= [1-p_i^x(1-p_i^x)^{\delta-1}][-\Delta s_i(\mu- \delta)], \nonumber 
\end{align}
\begin{align}
    & P_{\mathcal{X}_i, \mathcal{T}_i}  (\mathcal{X}_i(\mu-\delta)=0 \cap \mathcal{T}_i(\mu,\delta)=0) \nonumber  \\ &= P(\mathcal{T}_i(\mu,\delta)=0 |\mathcal{X}_i(\mu - \delta)=0)P(\mathcal{X}_i(\mu -\delta)=0) \nonumber  \\ &= 1+\Delta s_i(\mu- \delta). \nonumber
\end{align}

Let $\mathcal{W}_i(\mu)$ be the marginal distribution of $\mathcal{T}_i(\mu,\delta)$ over the set of  feasible delays, $\delta$, with its pmf being 
the 
probability of a random
individual acquiring
a positive test at time $\mu$: 
\begin{align}
    &P_{\mathcal{W}_i}(\mathcal{W}_i(\mu)=1) \nonumber \\
    & =\sum_{\delta =1}^{\mu -T_1} P_{\mathcal{X}_i,\mathcal{T}_i}(\mathcal{X}_i(\mu-\delta) \cap \mathcal{T}_i(\mu,\delta)=1) \nonumber \\
    & =  \sum_{\delta =1}^{\mu-T_1} p_i^x (1-p_i^x)^{\delta-1}[ -\Delta s_i(\mu-\delta)], \nonumber 
\end{align}
by combining~\eqref{eq:joint_pmf} and~\eqref{eq:zero}.
Therefore, the number of confirmed cases at time $k$ is calculated as \begin{align}
    C_i(k) &=  
    \E[\sum_{l=1}^{N_i} \mathcal{W}_i(k) ] \nonumber \\
    &=\sum_{l=1}^{N_i} \E[ \mathcal{W}_i(k) ] \label{eq:indicator} \\
    & =N_i \sum_{\delta=1}^{k-T_1} p_i^x (1-p_i^x)^{\delta-1}[ -\Delta s_i(k-\delta)], \label{eq:C_i_k}
\end{align}
where~\eqref{eq:indicator} holds because of the linearity of expectation and since the testing delays are i.i.d. 
Hence, by combining $c_i(k) =\frac{C_i(k)}{N_i}$ and~\eqref{eq:C_i_k} for all $k\in [T_1+1,T_2+1]$, the transfer matrix $\Phi(T_1, T_2)$ 
in~\eqref{eq:transfer_matrix} is written as 
\begin{align}
   &\Phi(T_1, T_2) \nonumber \\ &= \begin{bmatrix}
0 & 0  & 0 & 0 &\hdots & 0    \\
p_i^x  & 0  & 0 & 0 & \hdots & 0    \\
p_i^x (1-p_i^x) & p_i^x  & 0 & 0 & \hdots & 0   \\
p_i^x (1-p_i^x)^2 &p_i^x (1-p_i^x) &  p_i^x & 0 & \hdots & 0   \\
\vdots & \vdots & \ddots & \ddots & \ddots &\vdots \\
p_i^x(1-p_i^x)^{q-2} & p_i^x(1-p_i^x)^{q-3} &\hdots &\hdots & p_i^x & 0
\end{bmatrix},\label{eq:transfer_geo}
\end{align}
where $q = T_2-T_1+1$.
By combining~\eqref{eq:transfer_matrix},~\eqref{eq:transfer_geo}, 
we obtain that
\begin{equation}\label{eq:model_daily}
    c_i(k) = p_i^x (-\Delta s_i(k-1))+ (1-p_i^x)c_i(k-1),
\end{equation}
for all $k \in [T_1+1,T_2+1]$. Meanwhile, we set $c_i(k) =0$ for all $k \notin [T_1+1, T_2+1]$, since no testing occurs.
\begin{remark}
The proportion of daily confirmed cases $c_i(k)$ in \eqref{eq:model_daily} consists of two terms: the first term $p_i^x (-\Delta s_i(k-1))$ can be interpreted as an infected individual's urgency in obtaining a test, and the second term $(1-p_i^x)c_i(k-1)$ captures the unwillingness/unlikeliness of an infected individual acquiring a test.
\end{remark}



Finally, we relate the proportion of the daily number of recoveries, i.e., $d_i(k)$, with the underlying states. In the data collected, $d_i(k)$ corresponds to the change in the proportion of recovered individuals and the total number of known active cases $\mathbb{A}_i(k-1)$. We assume
\begin{equation}\label{eq:recovered_data}
    d_i(k) \sim  \mathtt{Bin}\Big(\frac{\mathbb{A}_i(k-1)}{N_i}, h\gamma_i(k-1)\Big).
\end{equation}
Namely,  each known active case recovers with healing rate $h \gamma_i(k-1)$. From~\cite{hota2020closed}, when the number of active
cases is large, $d_i(k)$ is approximately equal to $\frac{h\gamma_i(k-1) \mathbb{A}_i(k-1)}{N_i}$.

The above analysis links the collected data proportions with the underlying states of the system. If we acquire the parameter: $p_i^x$, 
we will be able to estimate the state systems as follows:


\begin{definition}\label{def:initial}
We assume that: $\widehat{x_i}(k) =\widehat{x_i}(0)$, $\widehat{r_i}(k) =\widehat{r_i}(0)$, and $\widehat{s_i}(k)=\widehat{s_i}(0)$, where $\widehat{x_i}(0), \widehat{r_i}(0), \widehat{s_i}(0) \in [0,1]$ for all $i\in [n], k< T_1$. Given the testing data set $\Omega_i(k)$ collected from time step $T_1+1$ to $T_2+1$, according to~\eqref{eq:model_daily}, we define the estimated proportion of new infections at node $v_i$ as
\begin{equation}\label{eq:infer_sus_geo}
    -\widehat{\Delta s}_i(k) =\frac{c_i(k+1)-(1-p_i^x)c_i(k)}{p_i^x } , k\in [T_1,T_2].
\end{equation}
\end{definition}
Notice that when $p_i^x =1$,~\eqref{eq:infer_sus_geo} becomes: 
\begin{equation*}
    -\widehat{\Delta s}_i(k) = c_i(k+1), k\in [T_1,T_2],
\end{equation*}
which can be interpreted as: every infected individual will be tested the day after being infected. Hence, the estimated change in proportion of infection on a given day $k$ exactly equals to the fraction of the number of positive cases on the next day $k+1$.

Moreover, we let $\widehat{\Delta r}_i(k) =0$ for $k =T_1$. Note that the following equality holds from the formulation of the SIR model:
\begin{align*}
    \widehat{\Delta s}_i(k) + \widehat{\Delta x}_i(k) + \widehat{\Delta r}_i(k) = 0.  
\end{align*}
We further define that 
\begin{IEEEeqnarray}{CC}
\IEEEyesnumber\label{eq:infer_SIR_st} 
  \widehat{s_i}(k) &= \widehat{s_i}(k-1)+ \widehat{\Delta s}_i(k), \nonumber \\
\widehat{x_i}(k) &= \widehat{x_i}(k-1)+ \widehat{\Delta x}_i(k), 
\\
\widehat{r_i}(k) &= \widehat{r_i}(k-1)+ \widehat{\Delta r}_i(k),\nonumber  \label{eq:infer_SIRsub3}
\end{IEEEeqnarray}
for $k \in [T_1,T_2]$.

According to~\eqref{eq:recovered_data} and~\cite{hota2020closed}, the change in the proportion of recovered individuals at node $v_i$ can be inferred as
\begin{equation}\label{eq:infer_recovered_change}
    \widehat{\Delta r}_i(k) = \frac{N_id_i(k)}{\mathbb{A}_i(k-1)}\widehat{x_i}(k-1), k \in [T_1+ 1, T_2],
\end{equation}
where $\widehat{x_i}(k-1)$ is calculated from \eqref{eq:infer_sus_geo} and~\eqref{eq:infer_SIR_st}. When $\mathbb{A}_i(k-1)=0$, we assume $\widehat{\Delta r}_i(k) =0$.

Therefore, if the testing data $\Omega_i(k)$ is available over an interval $k\in [T_1+1, T_2+1]$, we can estimate the states of the system by repetitively applying~\eqref{eq:infer_sus_geo},~\eqref{eq:infer_SIR_st}, and~\eqref{eq:infer_recovered_change} with the initial conditions, i.e., $\widehat{s_i}(0),$ $\widehat{x_i}(0),$ and $ \widehat{r_i}(0)$, assumed for the geometric distribution model. This addresses question~\ref{itm:second_a} in Section~\ref{sec:model}. 

\begin{assumption}\label{assum:inf}
We assume that $c_i(k) =0$ for all $k \in [T_1]\cup \{0\}$
and the initial inferred susceptible proportion is $\widehat{s_i}(0)$.
\end{assumption}
\begin{remark}
When estimating the system states, we first assume an initial condition for the system based on reality. We also assume that outside of the testing period, the proportion of positive cases collected is zero.
\end{remark}

\begin{proposition}\label{prop:estimation_error}
Under Assumption~\ref{assum:inf}, the error of the inference method in~\eqref{eq:infer_sus_geo}-
\eqref{eq:infer_recovered_change} 
at time $k$
is given by
\begin{equation}\label{eq:infer_error_bound}
    | \widehat{s_i}(k) -s_i(k) | = \left|\widehat{s_i}(0)-s_i(0)-\sum_{l=1}^{T_1-1}\Delta s_i(l)  
    \right|, 
\end{equation}
for all $k \geq T_1$.
\end{proposition}
\begin{proof}
From~\eqref{eq:dt_SIRsub1}, we first represent $s_i(k)$ by:
\begin{equation}\label{eq:sus_true_sum}
    s_i(k) = s_i(0) + \sum_{l=1}^k \Delta s_i(l).
\end{equation}
Now, we characterize $\widehat{s_i}(k)$:
\begin{align}
    \widehat{s_i}(k) &= \widehat{s_i}(0) + \sum_{l=1}^k \widehat{\Delta s}_i(l) \nonumber  \\ \label{eq:sum_hat_sus_2}
    &= \widehat{s_i}(0) + \sum_{l=T_1}^k \widehat{\Delta s}_i(l)   \\\label{eq:sum_hat_sus_3}
    & = \widehat{s_i}(0) - \frac{c_i(k+1)}{p_i^x}-\sum_{l=T_1+1}^k c_i(l),
\end{align}
where~\eqref{eq:sum_hat_sus_2} is written because $-\widehat{\Delta s}_i(l)=0$ for all $l\leq T_1-1$ in~\eqref{eq:infer_sus_geo}. We acquired~\eqref{eq:sum_hat_sus_3} through representing each $\widehat{\Delta s}_i(l)$, $l\geq T_1$ by~\eqref{eq:infer_sus_geo} and following Assumption~\ref{assum:inf}.   
By applying~\eqref{eq:model_daily}, we calculate the $\sum_{l=T_1+1}^k c_i(l)$ on the R.H.S. of~\eqref{eq:sum_hat_sus_3} as
\begin{align}\label{eq:sum_c_i(j)}
    &\sum_{l=T_1+1}^k c_i(l) =-p_i^x \sum_{l=T_1}^{k-1} \Delta s_i(l) +(1-p_i^x)\sum_{l=T_1+1}^{k-1} c_i(l) \\\label{eq:sum_c_i(j)_2}
    & \ \   \ \ =-p_i^x \sum_{l=T_1}^{k-1} \Delta s_i(l) +  (1-p_i^x) \Big[\sum_{l=T_1+1}^k c_i(l) -c_i(k) \Big],
\end{align}
since $\sum_{l=T_1+1}^{k-1} c_i(l)= \sum_{l=T_1+1}^k c_i(l) -c_i(k)$.
We can 
reorganize~\eqref{eq:sum_c_i(j)_2} and acquire:
\begin{equation}\label{eq:sum_c_i_clearversion}
    \sum_{l=T_1+1}^k c_i(l) = -\sum_{l=T_1}^{k-1} \Delta s_i(l) - \frac{1-p_i^x}{p_i^x} c_i(k).
\end{equation}
Hence, we replace $\sum_{l=T_1+1}^k c_i(l)$ on the R.H.S. of~\eqref{eq:sum_hat_sus_3} with~\eqref{eq:sum_c_i_clearversion} and obtain:
\begin{align}
    \widehat{s_i}(k) & = \widehat{s_i}(0) - \frac{c_i(k+1)}{p_i^x}+\sum_{l=T_1}^{k-1} \Delta s_i(l) 
    + \frac{1-p_i^x}{p_i^x} c_i(k)  \label{eq:s_hat_1}  \\ \label{eq:s_hat_2}
    &= \widehat{s_i}(0) + \sum_{l=T_1}^{k} \Delta s_i(l),
\end{align}
where~\eqref{eq:s_hat_2} 
follows from
writing $c_i(k+1)$ in~\eqref{eq:s_hat_1} 
as $p_i^x (-\Delta s_i(k))+(1-p_i^x)c_i(k)$, using~\eqref{eq:model_daily}. Therefore, we can calculate $  | \widehat{s_i}(k) -s_i(k) |$ by comparing~\eqref{eq:sus_true_sum} with~\eqref{eq:s_hat_2} and yield the result. 
\end{proof}
Prop.~\ref{prop:estimation_error} provides an analytical expression of the estimation error given the initial susceptible level assumed and the start testing time. Hence, Prop.~\ref{prop:estimation_error} solves question~\ref{itm:second_b} in Section~\ref{sec:model}.
\begin{corollary}\label{coro:est_error}
In Prop.~\ref{prop:estimation_error}, if we assume that $\widehat{s_i}(0)=1$, then we can write that $\widehat{s_i}(k)\geq s_i(k)$, for all $k \geq T_1$. Moreover, if we assume the inferred initial conditions 
correctly, i.e., $\widehat{s_i}(0)=s_i(0)$, and $T_1= 1$, then 
the  algorithm will estimate the susceptible state perfectly.
\end{corollary}
\begin{remark}
The result of Prop.~\ref{prop:estimation_error} consists of two parts: $\widehat{s_i}(0)-s_i(0)$ and $-\sum_{l=1}^{T_1-1}\Delta s_i(l)$. The first component depends on the difference between the inferred initial susceptible level and true initial susceptible level. The second component depends on the start testing date. Therefore, the accuracy of the estimation algorithm corresponds to the estimated initial condition and how early the testing data is collected. 
\end{remark}
\noindent
We will explore this error via simulations in Section \ref{sec:simulation}.

By estimating the proportion of infected individuals in a subpopulation of a network, we are able to acquire the estimation of the infection prevalence in the whole system. These inferred states provide an understanding of the epidemic and important factors for designing eradication schemes for infectious diseases.

\section{Distributed Eradication Strategy}\label{sec:eradication_strategies}
In this section, we propose two distributed strategies that employ the true states and the estimated states, respectively, and guarantee the eradication of the virus in at least exponential time. 





We propose the following healing rate to control the epidemic spread over the network:
\begin{equation}\label{eq:control_scheme_healing_rate}
    \widetilde{\gamma}_i(k) =s_i(k) \sum_{j=1}^n \beta_{ij}(k) +\epsilon_i,  \;\ i \in [n], 
\end{equation}
where 
$\epsilon_i>0$, for each $i \in [n]$.
This algorithm can be understood as boosting the healing rate of each subpopulation separately by providing effective medication, medical supplies, and/or healthcare workers.
\begin{theorem}\label{thm:control_datadriven}
Consider the system in~\eqref{eq:dt_SIR} and assume that
\begin{enumerate}
    \item $0\leq h \sum_j \beta_{ij}(k) < 1$, 
    $\forall i \in [n]$ and $\forall k \in  \mathbb{ Z}_{\geq 0}$,
    \item $B(k)$ is symmetric and irreducible $\forall k \in  \mathbb{ Z}_{\geq 0}$,
    \item $\exists \epsilon_i$ small enough that $h\widetilde{\gamma}_i(k) <1$, 
    $\forall i \in [n],
    k \in  \mathbb{ Z}_{\geq 0}$.
\end{enumerate}
Then the algorithm~\eqref{eq:control_scheme_healing_rate}  guarantees GES of the set of healthy states and $x(k)$ converges to $\mathbf{0}$ with at least an exponential rate.
\end{theorem}
\begin{proof}
By substituting \eqref{eq:control_scheme_healing_rate} into \eqref{eq:dt_SIR}, we obtain
\begin{align}\label{eq:infected_control}
    &x_i(k+1) = x_i(k) + \nonumber \\ & h\{s_i(k)\sum_{j=1}^n \beta_{ij}(k) x_j(k) - [s_i(k) \sum_{j=1}^n \beta_{ij}(k) +\epsilon_i]x_i(k)\}.
\end{align}
The state transition matrix of \eqref{eq:infected_control} can be written as 
\begin{equation}
    \widetilde{M}(k) = I+h[S(k)B(k) -(S(k)\text{diag}(B(k) \mathbf{1}_{n\times 1}) + \text{diag}(\epsilon_i))].
\end{equation}
For any $i
, j\in [n], j\neq i$, the entries of the $i$-th row of $\widetilde{M}(k)$ are 
\begin{equation}
    \widetilde{m}_{ii}(k) = 1-h[s_i(k)\sum\limits_{j\neq i}^n \beta_{ij}(k)+\epsilon_i],
\end{equation}
\begin{equation}
    \widetilde{m}_{ij}(k) = hs_i(k)\beta_{i,j}(k) ,
\end{equation}
which satisfies the following inequality
\begin{equation}
    \widetilde{m}_{ii}(k)+ \sum\limits_{j\neq i}^n\widetilde{m}_{ij}(k) \leq 1- h\min \{\epsilon_i\}, \forall i \in n.
\end{equation}
Therefore, by Gershgorin circle theorem, the spectral radius of $\widetilde{M}(k) $ is upper bounded by $1- h\min \{\epsilon_i\}$:
\begin{equation}
    \rho(\widetilde{M}(k)) \leq 1- h\min \{\epsilon_i\}. 
\end{equation}
Since we have $x(k+1) = \widetilde{M}(k) x(k)$ and $x(k) \geq 0$ for all $k$, we can write that $\lVert x(k+1) \rVert \leq[1- h\min \{\epsilon_i\}] \lVert x(k) \rVert$ for all $k$. Since $\epsilon_i >0, \forall i\in n$, we obtain that, for all $x_i(0)\in [0,1]^n$,
\begin{equation}\label{eq:converge}
    \lVert x(k) \rVert \leq [1- h\min \{\epsilon_i\}]^k \lVert x(0) \rVert  \leq e^{-k h\min \{\epsilon_i\}} \lVert x(0) \rVert,
\end{equation}
where the second inequality holds by
Bernoulli's inequality~\cite{carothers2000real}, \begin{equation}\label{eq:bernoulli}
    e^x= \lim_{n\rightarrow \infty}(1+\frac{x}{n})^n \geq 1+x.
\end{equation}
Hence, $x(k)$ converges to $\mathbf{0}$ with an exponential rate of at least $h\min \{\epsilon_i\}$.
Therefore, 
the set of healthy states is GES.
\end{proof}

\begin{remark}
The control strategy proposed in Theorem~\ref{thm:control_datadriven} can be interpreted as follows: if the healing rate of each subpopulation is appropriately increased according to its susceptible proportion, for example by distributing effective medication, medical supplies, and/or healthcare workers to each subpopulation, 
then the epidemic will be eradicated with at least an exponential rate. This theorem provides decision makers 
insight into, given sufficient resources, how to allocate 
medical supplies and healthcare workers
to different subpopulations
so that the epidemic can be eradicated quickly. Furthermore, Theorem~\ref{thm:control_datadriven} provides sufficient conditions for guaranteeing an exponentially decreasing $\lVert x(k) \rVert$ for all $k$ when the conditions apply. In other words, implementing the control strategy in Theorem~\ref{thm:control_datadriven} at full length will prevent the potential upcoming waves of the epidemic in the 2-norm sense of $x(k)$.
\end{remark}

Using the estimation results from Section~\ref{sec:inference}, we consider the following healing rate:
\begin{equation}\label{eq:control_scheme_healing_rate_inferred}
    \widehat{\widetilde{\gamma}}_i(k) =\widehat{s_i}(k) \sum_{j=1}^n \beta_{ij}(k) +\epsilon_i,  \;\ i \in [n], 
\end{equation}
where $\widehat{s_i}(k)$ is the estimated susceptible rate from~\eqref{eq:infer_SIR_st}.

\begin{corollary}\label{coro:control_inferred}
Consider the system in~\eqref{eq:dt_SIR} and assume that
\begin{enumerate}
    \item $0\leq h \sum_j \beta_{ij}(k) \leq 1$ , $\forall i \in [n]$ and $\forall k \in  \mathbb{ Z}_{\geq 0}$,
    \item $B(k)$ is symmetric and irreducible $\forall k \in  \mathbb{ Z}_{\geq 0}$.
    \item $\exists \epsilon_i$ small enough that $h\widehat{\widetilde{\gamma}}_i(k) <1$, $\widehat{s_i}(0)=1$
    $\forall i \in [n]$ and $\forall k \in  \mathbb{ Z}_{\geq 0}$.
\end{enumerate}
Then the algorithm~\eqref{eq:control_scheme_healing_rate_inferred} guarantees GES of the set of healthy states and $x(k)$ converges to $\mathbf{0}$ with at least an exponential rate.
\end{corollary}
\begin{proof}
Similar to the proof of Theorem~\ref{thm:control_datadriven}, we substitute \eqref{eq:control_scheme_healing_rate_inferred} into \eqref{eq:dt_SIR} and obtain the state transition matrix for $x_i(k)$:
\begin{equation}
    \widehat{\widetilde{M}}(k) = I+h[S(k)B(k) -(\widehat{S}(k)\text{diag}(B(k) \mathbf{1}_{n\times 1}) + \text{diag}(\epsilon_i))],
\end{equation}
where $\widehat{S}(k) =\text{diag}(\widehat{s_i}(k))$. 
For any $i
, j\in [n]$, $j\neq i$, the entries of the $i$-th row of $\widehat{\widetilde{M}}(k)$ satisfy:
\begin{align}
    \widehat{\widetilde{m}}_{ii}(k)+ \sum\limits_{j\neq i}^n\widehat{\widetilde{m}}_{ij}(k) 
    & \leq 1- h\min \{\epsilon_i\},
\end{align}
since from Corollary~\ref{coro:est_error} we know that when we assume that $\widehat{s_i}(0)=1$, 
we obtain $\widehat{s_i}(k)\geq s_i(k)$ for all $i\in [n]$. Consequently, by the Gershgorin circle theorem, we obtain that the spectral norm of $\widehat{\widetilde{M}}(k)$ is upper bounded by $ 1- h\min \{\epsilon_i\}$.

Therefore, by referring to~\eqref{eq:converge} and~\eqref{eq:bernoulli}, we acquire that 
the set of healthy states is GES.
\end{proof}

Theorem~\ref{thm:control_datadriven} (resp. Corollary~\ref{coro:control_inferred}) has proven that given the true (resp. estimated) susceptible state the distributed eradication strategy proposed eradicates the virus with at least an exponential rate. Therefore, question~\ref{itm:third_a} from Section~\ref{sec:model} has been addressed here.

In this section, we have presented two distributed eradication strategies based on the true and estimated system states. Both strategies ensure that the SIR epidemics converge to the sets of healthy states exponentially. 
We compare the two strategies with numerical simulations in Section~\ref{sec:simulation}, and study how a system will react if the eradication strategies are removed too early.

\section{Simulations} \label{sec:simulation}
\begin{figure}
\centering
\includegraphics[width=.18\textwidth]{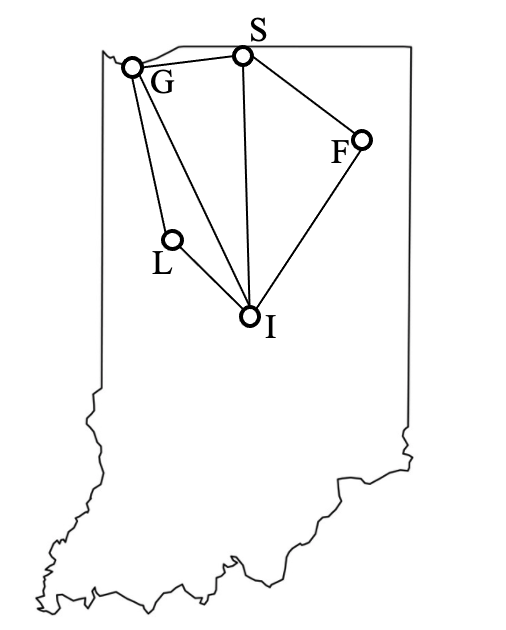}\includegraphics[width=.28\textwidth]{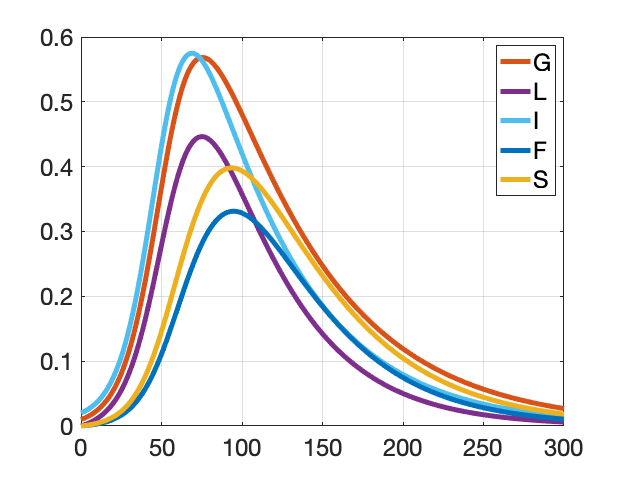}
\caption{Graph topology in the map of the state of Indiana~\cite{map2013} analyzed and the evolution of infected proportion in each city}
\label{fig:Graph_Topo_Indiana}
\end{figure}
In this section, we simulate a virus spreading over a static network
with 5 nodes in Fig~\ref{fig:Graph_Topo_Indiana} to illustrate our results. The nodes are modeled after the five metropolitan areas with a population over 150,000 in northern Indiana, U.S.: Gary (G), Lafayette (L), Indianapolis (I), Fort Wayne (F) and South Bend (S). Two nodes are neighbors if there is a major highway connecting 
them.
We set the initially infected proportion to be $0.02$ at node I and $0.01$ at node G and 0 elsewhere. The infection rates, healing rates, and the size of each subpopulation 
are static and presented in Table.~\ref{table:parameters}. 
The evolution of the infected proportion 
for each city
is shown in Fig.~\ref{fig:Graph_Topo_Indiana}.

\begin{table}[h!]
\centering
\begin{tabular}{|c | c c c c c|} 
 \hline
 $\beta_{ij}$    & G & L & I & F & S \\ [0.5ex] 
 \hline
 G   & 0.08      & 0.15          & 0.24   & 0  & 0.06  \\
 L & 0.15        & 0.12        & 0.13    & 0   & 0   \\
 I       & 0.24       & 0.13   & 0.25  & 0.05  & 0.04  \\
 F        & 0       & 0          & 0.05   & 0.11 & 0.15   \\
 S       &  0.06      & 0          & 0.04   & 0.14  & 0.09  \\ \hline\hline
 $\gamma_i$       & 0.075       & 0.115          & 0.085   & 0.125  & 0.1  \\
 $N_i$       & 500000       & 160000          & 900000   & 350000  & 300000  \\[1ex] 
 \hline
\end{tabular}
\caption{Network Parameters of Fig.~\ref{fig:Graph_Topo_Indiana}}
\label{table:parameters}
\end{table}



Considering the stochastic framework,
we simulate testing data using~\eqref{eq:model_daily} 
and~\eqref{eq:recovered_data}, with $p_i^x=0.2,$ 
$\forall i \in \{\text{G, L, I, F, S}\}$ from $T_1=6$ to $T_2=300$.
The number of daily and cumulative confirmed cases and removed (recovered) cases over time at node L are shown in Fig.~\ref{fig:Cases_Indiana_geo_p=0.3}. 
When $k\geq 80$, the proportion of infected individuals at node L begins to decrease in Fig.~\ref{fig:Graph_Topo_Indiana}, which leads to the decline of the number of active cases in Fig.~\ref{fig:Cases_Indiana_geo_p=0.3}. 

\begin{figure}
\centering
\begin{overpic}[width = 0.48\columnwidth]{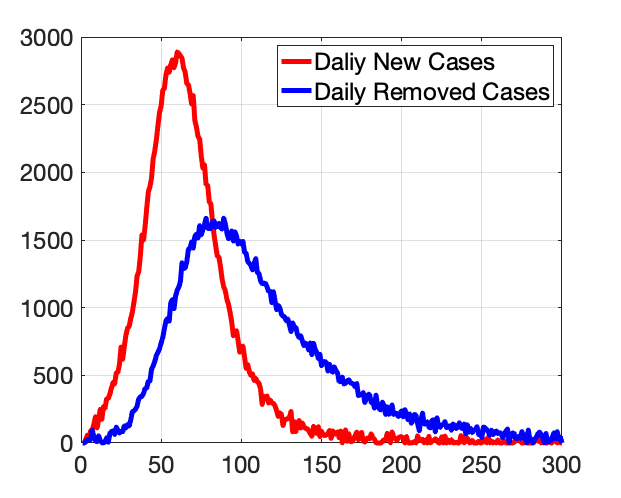}
     \put(-4,28){{\parbox{0.75\linewidth}\footnotesize \rotatebox{90}{\footnotesize$C_{\text{L}}$, $D_{\text{L}}$}
     }}
     \put(50,-3){\footnotesize{\parbox{0.75\linewidth}\footnotesize $k$
     }}\normalsize
   \end{overpic}
\begin{overpic}[width =  0.48\columnwidth]{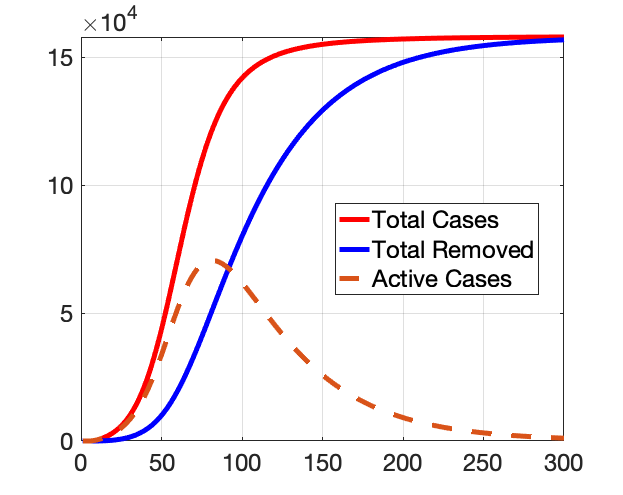}
     \put(-1,23){{\parbox{0.75\linewidth}\footnotesize \rotatebox{90}{\small$\mathbb{C}_{\text{L}}$, $\mathbb{D}_{\text{L}}$, $\mathbb{A}_{\text{L}}$}
     }}\normalsize
     \put(50,-3){\footnotesize 
    $k$
     }
   \end{overpic}
\caption{Simulated daily and cumulative number of cases for node L with $p_i^x=0.2$}
\label{fig:Cases_Indiana_geo_p=0.3}
\end{figure}

We now use the method proposed in Section~\ref{sec:inference} to estimate the susceptible proportion at node I. We assume that the initial condition of the recovered state is $\widehat{r}_{\text{I}}(0)=0$. Hence, the initial infected state is written as: $\widehat{x}_{\text{I}}(0) =1-\widehat{s}_{\text{I}}(0)$. In Fig.~\ref{fig:Estimation_error_sus}, 
we plot the absolute value of the estimation error of the susceptible state at $k=100$ versus the start testing time $T_1$ and initial condition assumed, $\widehat{s}_{\text{I}}(0)$. It can be seen in Fig.~\ref{fig:Estimation_error_sus} (top) that the estimation error increases linearly with the initial susceptible level assumed. When the initial condition is assumed correctly for node I, with a later start testing date, the estimation error at $k=100$ builds up from $0$ to $r_{\text{I}}(k)$ eventually. The increase in the estimation error with $T_1$ signifies the importance of an early testing during an outbreak: with appropriate initial conditions assumed, we should initiate testing as quickly as possible to improve the accuracy of the state estimation. Meanwhile, Fig.~\ref{fig:Estimation_error_sus} (bottom) indicates that with a later start testing date, we must assume a lower initial susceptible level accordingly to achieve accurate estimation. The intuition behind this finding is that 
since, 
by Definition \ref{def:initial}, $\widehat{s_i}(k)=\widehat{s_i}(0)$, for all $k< T_1$, the lower initial condition can compensate for missed tests from $k\in [0,T_1-1]$, captured by the last term in \eqref{eq:infer_error_bound}. However, guessing $\widehat{s_i}(0)$ correctly, namely $\widehat{s_i}(0)= s_i(0)+\sum_{l=1}^{T_1-1}\Delta s_i(l)$, for $T_1>0$ is quite difficult. 
Additionally, if we assume that $\widehat{s_i}(0)=1$, the estimated $\widehat{s_i}(k)$ is always larger than the true susceptible state in Fig.~\ref{fig:Estimation_error_sus}. The overestimation of susceptible level encourages us to design a stronger strategy to eradicate the virus, as will be seen in the subsequent simulations.

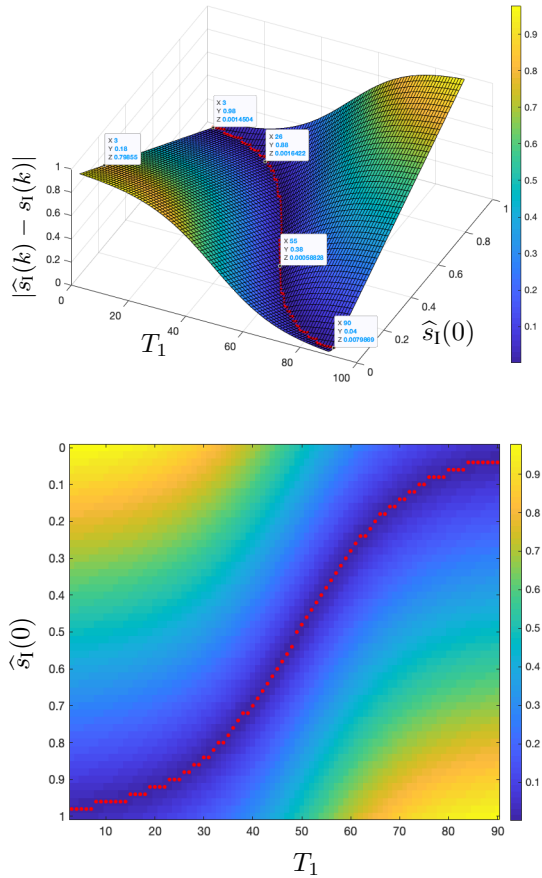
\begin{figure}
\centering
\begin{overpic}[width = 0.9 \columnwidth]{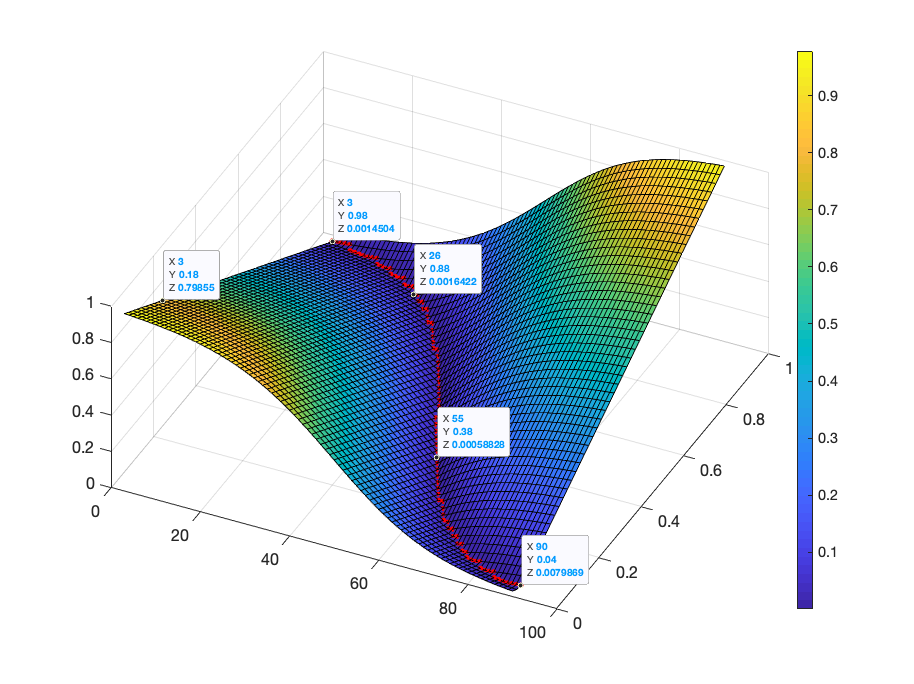}

    \put(2,19){{\parbox{0.75\linewidth}\tiny \rotatebox{90}{\small$| \widehat{s}_{\text{I}}(k) -s_{\text{I}}(k) | $}
     }}\normalsize
     \put(24,10){{\parbox{0.75\linewidth}\small 
     $T_1$ }
     }\normalsize
    \put(72,12){{\parbox{0.75\linewidth}\small 
     $\widehat{s}_{\text{I}}(0)$ }
     }\normalsize
   \end{overpic}
   \begin{overpic}[width = 0.91 \columnwidth]{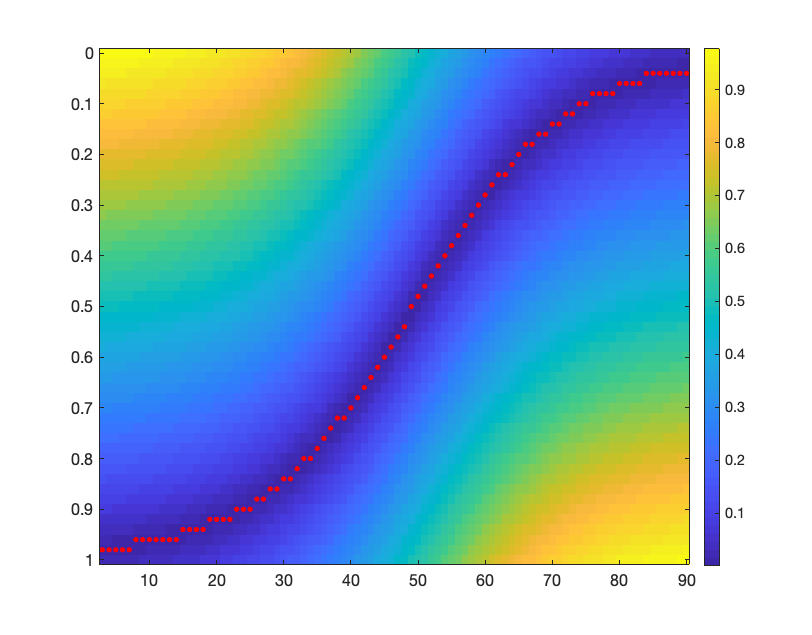}

    \put(2,34){{\parbox{0.75\linewidth}\small \rotatebox{90}{$\widehat{s}_{\text{I}}(0)$}
     }}\normalsize
    \put(50,0){{\parbox{0.75\linewidth}\small 
     $T_1$ }
     }\normalsize
   \end{overpic}
\caption{The absolute value of susceptible state estimation error at $k=100$ with respect to start testing date and the initial susceptible level assumed at node I, where the red points represent the estimation error: $| \widehat{s}_{\text{I}}(k) -s_{\text{I}}(k)| <0.01$. Both of the plots are illustrations of Prop.~\ref{prop:estimation_error}}
\label{fig:Estimation_error_sus}
\end{figure}

We simulate three scenarios over the network in Fig.~\ref{fig:Graph_Topo_Indiana} with the parameters of Table.~\ref{table:parameters}: no control, the distributed eradication strategy in~\eqref{eq:control_scheme_healing_rate}, and the  
distributed eradication strategy utilizing estimated states in~\eqref{eq:control_scheme_healing_rate_inferred}. The inferred states were produced by the algorithm in Section~\ref{sec:inference} with $p_i^x=0.5$ $\forall i \in \{\text{G, L, I, F, S}\}$.
The average states for each scenario are plotted in Fig.~\ref{fig:Control_strategies}. It can be seen that both eradication strategies are able to eliminate the virus at a much higher speed than with no control. Furthermore, when $k\geq 200$, the healthy states with the two eradication strategies applied achieve higher susceptible fractions than the healthy state without control.  
We can interpret the higher susceptible proportion as fewer individuals in the network becoming sick during the entire outbreak. 
The control algorithm from~\eqref{eq:control_scheme_healing_rate_inferred} converges to a healthy state faster than
the algorithm in~\eqref{eq:control_scheme_healing_rate}, and
both eradication strategies prevent resurgences of the virus over the network.
In Fig.~\ref{fig:Lift_Control_strategies}, we remove both eradication strategies when $k=50$ and $k=100$ and do not reinstate them. It can be seen that both of the infection curves rise up when $k\geq 50$ (resp. $k\geq 100$), and reach peaks before they slowly die out. Fig.~\ref{fig:Lift_Control_strategies} can be interpreted as removing the allocation of resources and healthcare workers from a subpopulation too early during a pandemic, resulting in the increase in infection level and a potential outbreak. In Fig.~\ref{fig:Control_strategies_inter}, we only enforce our eradication strategies within time interval: $k\in [20,50]$ and $k \in [20,150]$, respectively. We can see that although the control strategies reduce the infection level significantly, a resurgence of the outbreak occurs immediately upon the removal of the eradication strategies. Hence, policy makers are suggested to enforce the eradication strategies during the entire outbreak to avoid the upcoming wave of epidemic.

\begin{figure}
\begin{overpic}[width = 0.49 \columnwidth]{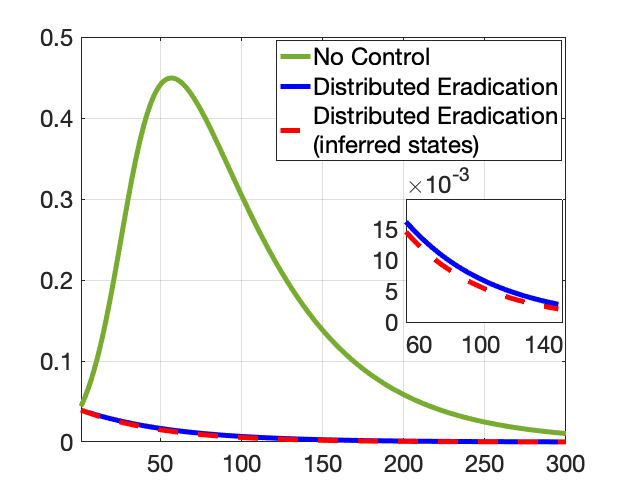}

    \put(-3,23){{\parbox{0.75\linewidth}\small \rotatebox{90}{\small$\frac{1}{n}\sum_i^n x_i(k)$}
     }}\normalsize
     \put(50,-3){{\parbox{0.75\linewidth}\small $k$
     }}\normalsize
   \end{overpic}
   \begin{overpic}[width = 0.49 \columnwidth]{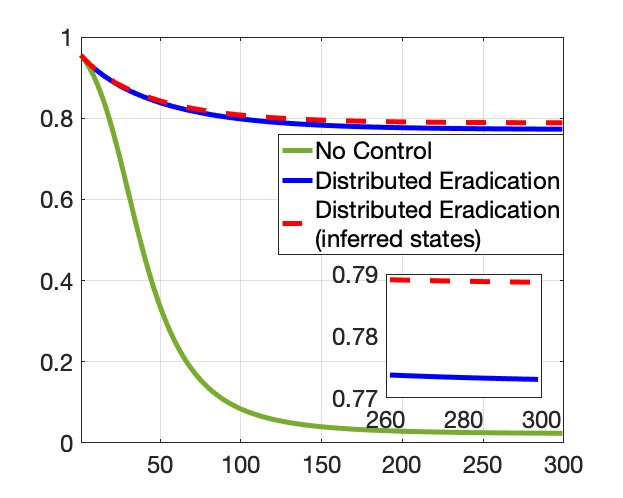}


    \put(-3,23){{\parbox{0.75\linewidth}\small \rotatebox{90}{\small$\frac{1}{n}\sum_i^n s_i(k)$}
     }}\normalsize
     \put(50,-3){{\parbox{0.75\linewidth}\small $k$
     }}\normalsize
   \end{overpic}
\caption{Average system states over time}
\label{fig:Control_strategies}
\end{figure}

\begin{figure}
\begin{overpic}[width = 0.49 \columnwidth]{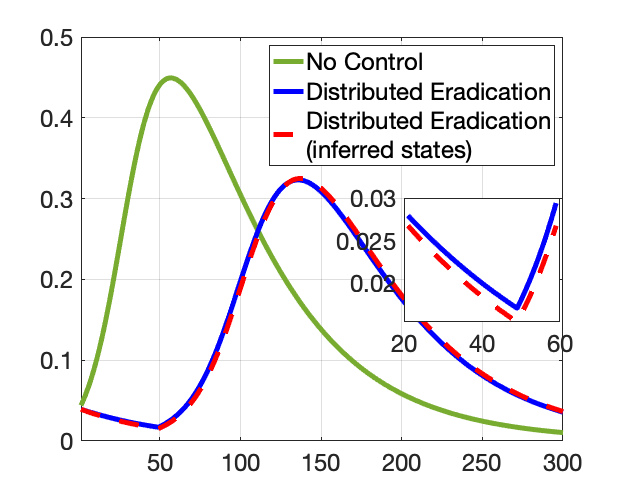}

    \put(-4,20){{\parbox{0.75\linewidth}\small \rotatebox{90}{\small$\frac{1}{n}\sum_i^n x_i(k)$}
     }}\normalsize
     \put(50,-1.5){{\parbox{0.75\linewidth}\small $k$
     }}\normalsize
   \end{overpic}
   \begin{overpic}[width = 0.49 \columnwidth]{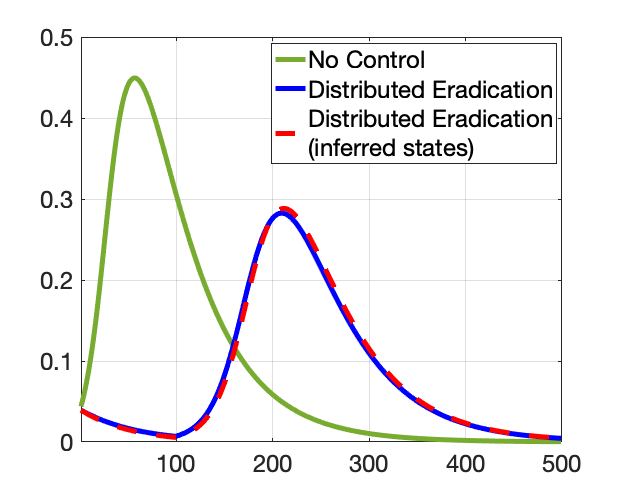}

    \put(-4,20){{\parbox{0.75\linewidth}\small \rotatebox{90}{\small$\frac{1}{n}\sum_i^n x_i(k)$}
     }}\normalsize
     \put(50,-1.5){{\parbox{0.75\linewidth}\small $k$
     }}\normalsize
   \end{overpic}
\caption{Average infection proportion of the virus over time with the eradication strategies enforced at $k\in [0,50]$ (left) and at $k \in [0,100]$ (right)}
\label{fig:Lift_Control_strategies}
\end{figure}


\begin{figure}
\begin{overpic}[width = 0.49 \columnwidth]{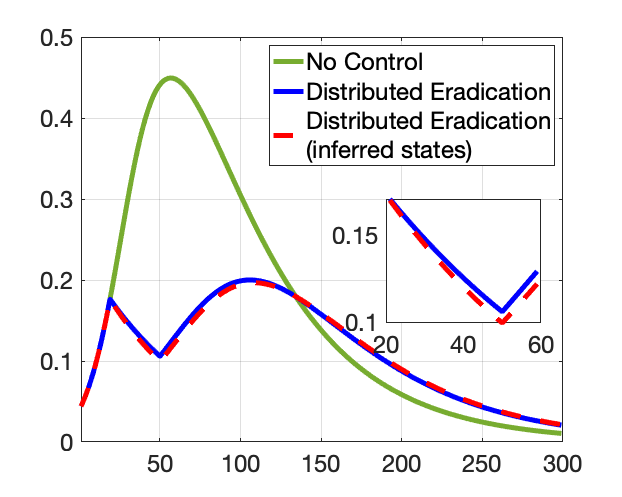}

    \put(-4,20){{\parbox{0.75\linewidth}\small \rotatebox{90}{\small$\frac{1}{n}\sum_i^n x_i(k)$}
     }}\normalsize
     \put(50,-1.5){{\parbox{0.75\linewidth}\small $k$
     }}\normalsize
   \end{overpic}
   \begin{overpic}[width = 0.49 \columnwidth]{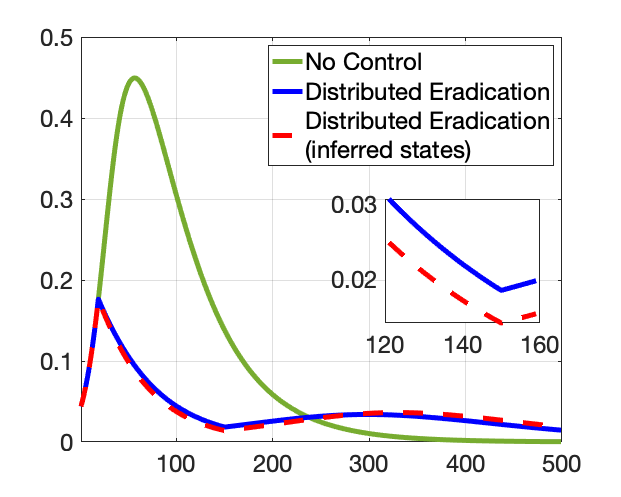}

    \put(-4,20){{\parbox{0.75\linewidth}\small \rotatebox{90}{\small$\frac{1}{n}\sum_i^n x_i(k)$}
     }}\normalsize
     \put(50,-1.5){{\parbox{0.75\linewidth}\small $k$
     }}\normalsize
   \end{overpic}
\caption{Average infection proportion of the virus over time with both of the eradication strategies imposed at $k\in [20,50]$ (left) and at $k\in [20,150]$ (right)}
\label{fig:Control_strategies_inter}
\end{figure}


\section{Conclusion}\label{sec:conclusion}
This paper studied the stability, inference and control of discrete time, time-varying SIR epidemics over networks. We established the sufficient condition for GES of the set of healthy states. In addition, we proposed a stochastic framework for estimating the 
underlying epidemic states
from collected testing data. We provided analytic expressions for the error of the estimation algorithm. 
We also proposed two distributed control strategies that are able to eradicate the virus in at least exponential time. 
The control strategies provide 
insights for decision makers on how to eliminate an ongoing outbreak.

In future work, we plan to study the stability and control of models with more states than SIR such as SEIRS (susceptible-exposed-infected-recovered-susceptible) and SAIR (susceptible-asymptomatic-infected-recovered) as they can possibly capture the characteristics of COVID-19 better. In our stochastic testing framework, we did not consider the existence of 
inaccurate testing kits,
which appear frequently and cause confusion for 
policy 
makers. Hence, we plan to include false positive/negative test results into our testing and estimation model and investigate the new model's estimation accuracy in the future. Furthermore, we aim to apply our model on the real data to identify the system parameters. 




\section*{Acknowledgment}

The authors would like to thank Ashish Hota (IIT Kharagpur) and Baike She (Purdue University) for useful conversations related to 
Section~\ref{sec:inference}.




%




\normalem
\bibliographystyle{IEEEtran}
\bibliography{reference}

  \appendix
\subsection{Proof of Theorem.~\ref{thm:GES}}

\begin{proof}
By Assumptions \ref{assume:one} and \ref{assume:two}, $M(k)$ is nonnegative. Therefore, from Lemma~\ref{lemma:diagonal}, for all $k \in \mathbb{Z}_{\geq 0}$, there exists a diagonal matrix $Q(k)\succ 0$ such that $M^\top(k) Q(k+1) M(k) - Q(k)\prec 0$. 

Consider the following Lyapunov function $V(k,x) = x^\top Q(k)x$. Since for all $k \in \mathbb{Z}_{\geq 0}$, $Q(k)$ is diagonal and 
positive definite, it can be written that
$x^\top Q(k)x>0$, for all $x\neq \mathbf{0}$. 
Therefore, $V(k,x)>0$ for all $k\in \mathbb{Z}_{\geq 0}$, $x\neq \mathbf{0}$. 
Additionally,
all the eigenvalues of $Q(k)$ are real and positive. By applying Rayleigh-Ritz Quotient Theorem~\cite{horn2012matrix}, we obtain
\begin{equation}
  \lambda_{\text{min}}(Q(k))I \leq Q(k) \leq  \lambda_{\text{max}}(Q(k))I,    
\end{equation}
which implies
\begin{equation}\label{eq:sigma2}
    \sigma_1\| x\|^2 \leq V(k,x) \leq  \sigma_2\| x\|^2,
\end{equation}
where 
$\sigma_1 = \min_{k\in \mathbb{Z}_{\geq 0}}\lambda_{\text{min}}(Q(k))$ and $\sigma_2= \max_{k\in \mathbb{Z}_{\geq 0}} \lambda_{\text{max}}(Q(k))$, with $\sigma_1, \sigma_2>0$,
for all $k \in \mathbb{Z}_{\geq 0}$.

Now we turn to $\Delta V(k,x)$. For $x\neq 0$ and for each $k\in \mathbb{ Z}_{\geq 0}$, using \eqref{eq:SIR_dynamic} and \eqref{eq:M}-\eqref{eq:Mhat}, we can write 
\begin{align}
    &\Delta V(k,x) \nonumber \\ &= 
    x^\top \hat{M}(k)^\top Q(k+1)\hat{M}(k)x- x^\top Q(k) x \nonumber \\ 
    &= 
    x^\top [M(k)^\top  Q(k+1)M(k)- Q(k)]x \nonumber \\ & \ \  \ \ -2h x^\top B^\top(k) (I-S(k))Q(k+1)M(k)x\nonumber \\ 
    & \ \  \ \  +h^2 x^\top B^\top(k) (I-S(k))Q(k+1)(I-S(k))B(k)x.\label{eq:deltaV}
\end{align}
Note that the second and third term of \eqref{eq:deltaV} can be reorganized as 
\begin{align}
    &x^\top [-2h B^\top(k) (I-S(k))Q(k+1)M(k)\nonumber \\
    &+h^2 B^\top(k) (I-S(k))Q(k+1)(I-S(k))B(k) ] x \nonumber \\
    &= x^\top \{h B^\top(k) (I-S(k))Q(k+1)\nonumber \\
    & \ \  \ \ [-2M(k)+ h(I-S(k)) B(k)]\}x\nonumber \\
    &= x^\top \{h B^\top(k) (I-S(k))Q(k+1)\nonumber \\ \label{eq:second_third_term}
    & \ \  \ \ [-2(I -h\Gamma(k))-h(I+S(k))B(k) ]\}x\leq 0, 
\end{align}
where the last equality follows from \eqref{eq:M}, and the inequality follows from Assumptions \ref{assume:one} and \ref{assume:two} and Lemma~\ref{lemma:one}. 
Thus, by applying~\eqref{eq:second_third_term} into~\eqref{eq:deltaV}, we obtain that
\begin{equation}
    \Delta V(k,x) \leq x^\top [M(k)^\top  Q(k+1)M(k)- Q(k)]x.
\end{equation}
From Lemma~\ref{lemma:diagonal}, we know that $[M(k)^\top  Q(k+1)M(k)- Q(k)]$ is negative definite, and $[Q(k)-M(k)^\top  Q(k+1)M(k)]$ is positive definite. Therefore, we obtain $\lambda_{\text{max}}[M(k)^\top  Q(k+1)M(k)- Q(k)] =-\lambda_{\text{min}}[Q(k)-M(k)^\top  Q(k+1)M(k)]$. By applying Rayleigh-Ritz Quotient Theorem we can write 
\begin{equation}\label{eq:sigma3}
    \Delta V(k,x) \leq - \sigma_3 \| x\|^2,
\end{equation}
where $\sigma_3 = \max_{k\in \mathbb{Z}_{\geq 0}} \lambda_{\text{min}}[Q(k)-M(k)^\top  Q(k+1)M(k)]$, with $\sigma_3>0$ for all $k \in \mathbb{Z}_{\geq 0}$.

Therefore, from~\eqref{eq:sigma2},~\eqref{eq:sigma3} and  Lemma~\ref{lemma:GES} the set of healthy states of \eqref{eq:dt_SIR} is GES, proving the theorem.
\end{proof}

\subsection{Proof of Corollary~\ref{coro:rate_exp}}
\begin{proof}
From Lemma~\ref{lemma:rate_GES}, \eqref{eq:sigma2}, and \eqref{eq:sigma3}, the rate of convergence is upper bounded by $\sqrt{1-\frac{\sigma_3}{\sigma_2}}$.  
The next step is to show that the rate is well defined, which is $\sqrt{1-\frac{\sigma_3}{\sigma_2}}  \in [0,1)$. Since $\sigma_2>0$ and $\sigma_3>0$, we only need to prove that $\sigma_2 \geq \sigma_3$.

Note that for all $k \in \mathbb{Z}_{\geq 0}$, $Q(k)$ and $M(k)^\top  Q(k+1)M(k)$ are both symmetric. Therefore, by applying Weyl's inequalities from~\cite{horn2012matrix} to $[Q(k)-M(k)^\top  Q(k+1)M(k)]$, we obtain that, for all $k \in \mathbb{Z}_{\geq 0}$,
\begin{align}\label{eq:Weyl}
    & \lambda_i[Q(k)- M(k)^\top Q(k+1) M(k) ] \nonumber \\
    & \leq \lambda_i(Q(k))-\lambda_i [M(k)^\top Q(k+1) M(k) ].
\end{align}
We compare the LHS of~\eqref{eq:Weyl} with $\sigma_3$ and the RHS of~\eqref{eq:Weyl} with $\sigma_2$ to yield
\begin{align}
     \sigma_3 & \leq \lambda_i[Q(k)- M(k)^\top Q(k+1) M(k) ], \\
     \sigma_2 &\geq \lambda_i(Q(k))-\lambda_i [M(k)^\top Q(k+1) M(k)],
     \label{eq:RHS_Weyl}
\end{align}
where~\eqref{eq:RHS_Weyl} holds because $M(k)^\top Q(k+1) M(k)$ is positive semidefinite for all $k \in \mathbb{Z}_{\geq 0}$. 
Hence, $\sigma_2\geq \sigma_3$ and the rate of convergence is well defined.
\end{proof}
\end{document}